%% file: main.tex
\documentclass[runningheads]{llncs}
\usepackage{graphicx, amsmath, amssymb, color} 
\usepackage{enumitem}
\usepackage[ruled,vlined]{algorithm2e}
\usepackage{url}
\usepackage[capitalise]{cleveref}

\newcommand{\msize}[1]{{\left|#1\right|}}

\newcommand{\dist}{{\rm dist}}

\newcommand{\calC}{{\cal C}}

\newcommand{\Col}{{\mbox{\sc Col}}}
\newcommand{\mA}{\alpha}
\newcommand{\mAi}{\alpha^{-1}}
\newcommand{\mB}{\beta}
\newcommand{\mBi}{\beta^{-1}}
\newcommand{\mCo}{\gamma_1}
\newcommand{\mCoi}{\gamma_1^{-1}}
\newcommand{\mCt}{\gamma_2}
\newcommand{\mCti}{\gamma_2^{-1}}
\newcommand{\mDo}{\delta_1}
\newcommand{\mDt}{\delta_2}
\newcommand{\mT}{\tau}

\title{Cyclic Shift Problems on Graphs\thanks{This work is partially supported by KAKENHI grant numbers 17H06287 and 18H04091.}}
\author{Kwon Kham Sai
\and 
Ryuhei Uehara
\and 
Giovanni Viglietta
}
\institute{School of Information Science, 
Japan Advanced Institute of Science and Technology (JAIST), Japan.
\email{\{saikwonkham,uehara,johnny\}@jaist.ac.jp}}

\begin{document}
\maketitle

\begin{abstract}
We study a new reconfiguration problem inspired by classic mechanical puzzles: a colored token is placed on each vertex of a given graph; we are also given a set of distinguished cycles on the graph. We are tasked with rearranging the tokens from a given initial configuration to a final one by using cyclic shift operations along the distinguished cycles. We first investigate a large class of graphs, which generalizes several classic puzzles, and we give a characterization of which final configurations can be reached from a given initial configuration. Our proofs are constructive, and yield efficient methods for shifting tokens to reach the desired configurations. On the other hand, when the goal is to find a shortest sequence of shifting operations, we show that the problem is NP-hard, even for puzzles with tokens of only two different colors.
\end{abstract}

\section{Introduction}
\label{sec:1}
\input{intro}

\section{Preliminaries}
\label{sec:2}
\input{pre}

\section{Algebraic Analysis of the Puzzles}
\label{sec:3}
\input{puz}

\section{NP-Hardness for Puzzles with Two Colors}
\label{sec:4}
\input{np}

\input{appendix}

\end{document}

%% file: intro.tex
Recently, variations of reconfiguration problems have been attracting much interest, and 
several of them are being studied as important fundamental problems in theoretical computer science~\cite{nishimura}. Also, many real puzzles which can be modeled as reconfiguration problems have been invented and proposed by the puzzle community, such as the 15-puzzle and Rubik's cube. Among these, we focus on a popular type of puzzle based on cyclic shift operations: see~\cref{fig:puzzle}. In these puzzles, we can shift some elements along predefined cycles as a basic operation, and the goal is to rearrange the pieces into a desired pattern.

\begin{figure}
\centering
\includegraphics[height=3.75cm]{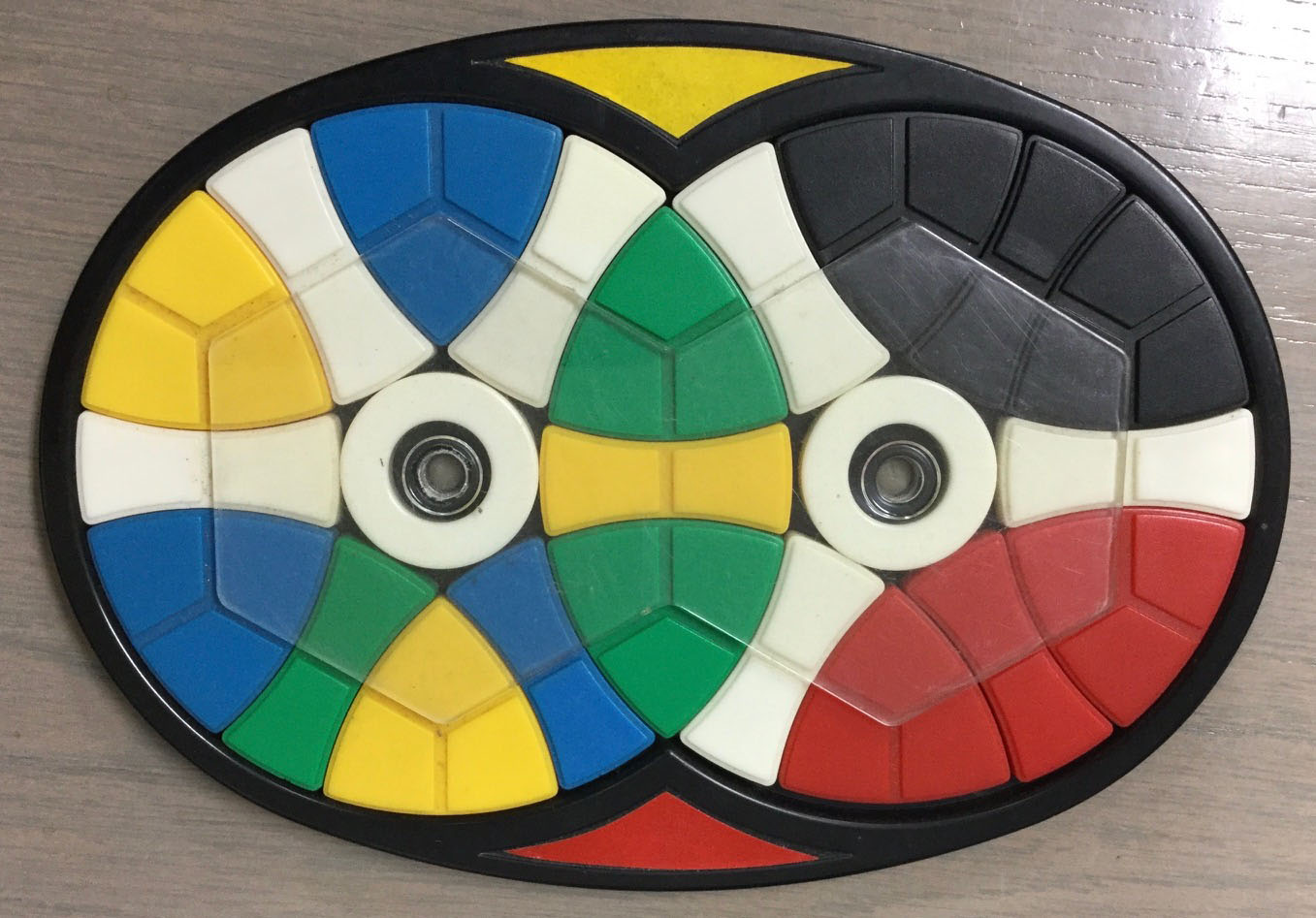}
\hspace{5mm}
\includegraphics[height=3.75cm]{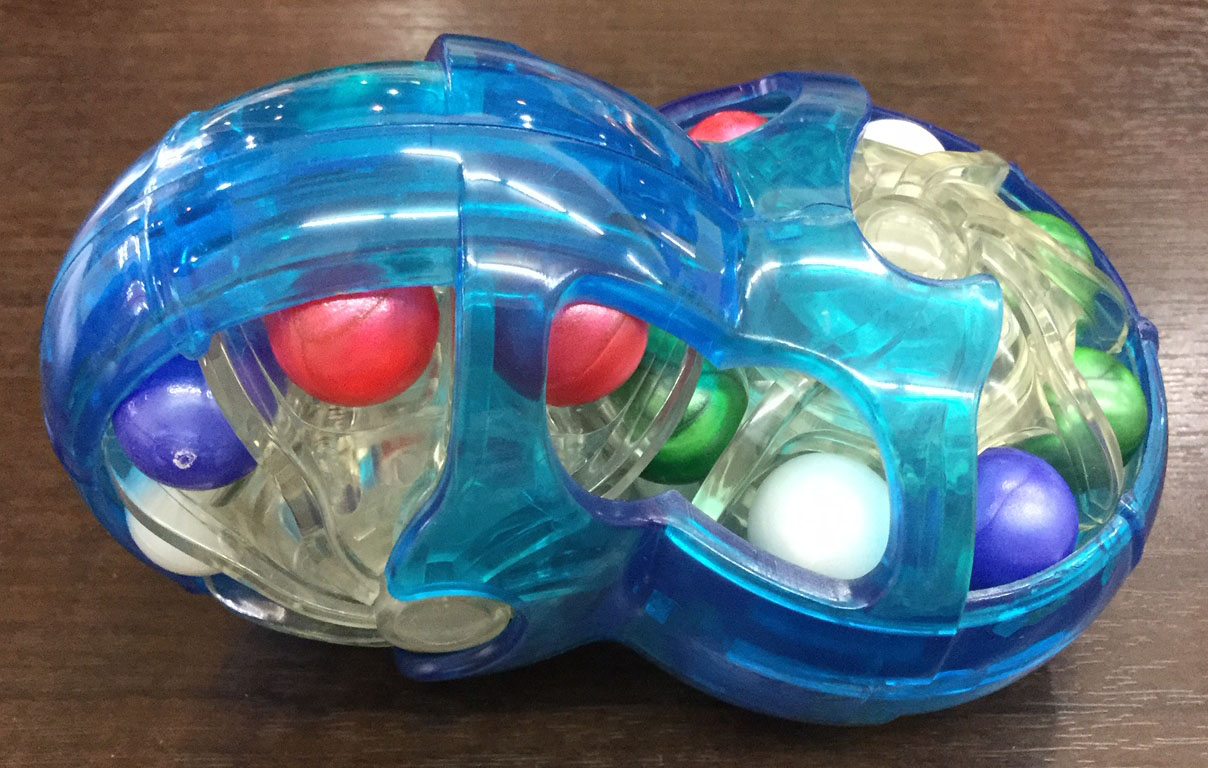}
\caption{Commercial cyclic shift puzzles: Turnstile (left) and Rubik's Shells (right)}
\label{fig:puzzle}
\end{figure}

In terms of reconfiguration problems, this puzzle can be modeled as follows.
The input of the problem is a graph $G=(V,E)$, a set of colors $\Col=\{1,2,\ldots,c\}$, and one colored token on each vertex in $V$. We are also given a set $\calC$ of cycles of $G$. The basic operation on $G$ is called ``shift'' along a cycle $C$ in $\calC$, and it moves each token located on a vertex in $C$ into the next vertex along $C$. This operation generalizes the token swapping problem, which was introduced by Yamanaka et al.~\cite{Yamanakaetal2018}, and has been well investigated recently. Indeed, when we restrict each cycle in $\calC$ to have length two (each cycle would correspond to an edge in $E$), the cyclic shift problem is equivalent to the token swapping problem.

In the mathematical literature, the study of permutation groups and their generators has a long history. An important theorem by Babai~\cite{babai} states that the probability that two random permutations of $n$ objects generate either the symmetric group $S_n$ (i.e., the group of all permutations) or the alternating group $A_n$ (i.e., the group of all even permutations) is $1-1/n+O(n^2)$. However, the theorem says nothing about the special case where the generators are cycles.

In~\cite{heath}, Heath et al.\ give a characterization of the permutations that, together with a cycle of length $n$, generate either $A_n$ or $S_n$, as opposed to a smaller permutation group. On the other hand, in~\cite{jones}, Jones shows that $A_n$ and $S_n$ are the only finite primitive permutation groups containing a cycle of length $n-3$ or less. However, his proof is non-constructive, as it heavily relies on the classification of finite simple groups (and, as the author remarks, a self-contained proof is unlikely to exist). In particular, no non-trivial upper bound is known on the distance of two given permutations in terms of a set of generators.

The computational complexity of related problems has been studied, too. It is well known that, given a set of generators, the size of the permutation group they generate is computable in polynomial time. Also, the inclusion of a given permutation $\pi$ in the group is decidable in polynomial time, and an expression for $\pi$ in terms of the generators is also computable in polynomial time~\cite{furst}.

In contrast, Jerrum showed that computing the distance between two given permutations in terms of two generators is PSPACE-complete~\cite{jerrum}. However, the generators used for the reduction are far from being cycles.

In this paper, after giving some definitions (\cref{sec:2}), we study the configuration space of a large class of cyclic shift problems which generalize the puzzles in \cref{fig:puzzle} (\cref{sec:3}). We show that, except for one special case, the permutation group generated by a given set of cycles is $S_n$ if at least one of the cycles has even length, and it is $A_n$ otherwise. This result is in agreement with Babai's theorem~\cite{babai}, and shows a similarity with the configuration space of the (generalized) 15-puzzle~\cite{wilson}. Moreover, our proofs in \cref{sec:3} are constructive, and yield polynomial upper bounds on the number of shift operations required to reach a given configuration. This is contrasted with \cref{sec:4}, where we show that finding a shortest sequence of shift operations to obtain a desired configuration is NP-hard, even for puzzles with tokens of only two different colors.

%% file: pre.tex
Let $G=(V,E)$ be a finite, simple, undirected graph, where $V$ is the vertex set, with $n=\msize{V}$, and $E$ is the edge set.
Let $\Col=\{1,2,\ldots,c\}$ be a set of colors, where $c$ is a constant.
A \emph{token placement} for $G$ is a function $f\colon V\to \Col$:
that is, $f(v)$ represents the color of the token placed on the vertex $v$.
Without loss of generality, we assume $f$ to be surjective. 

Let us fix a set $\calC$ of cycles in $G$ (note that $\calC$ does not necessarily contain all cycles of $G$). 
Two distinct token placements $f$ and $f'$ of $G$ are \emph{adjacent} with respect to $\calC$ if the following
two conditions hold:
(1) there exists a cycle $C=(v_1, v_2, \dots, v_j)$ in $\calC$ such that
       $f'(v_{i+1})=f(v_i)$ and $f'(v_1)=f(v_j)$ or 
       $f'(v_{i})=f(v_{i+1})$ and $f'(v_j)=f(v_1)$ for $1\le i\le j$, and 
(2) $f'(w)=f(w)$ for all vertices $w\in V\setminus \{v_1,\ldots,v_i\}$.
In this case, we say that $f'$ is obtained from $f$ by \emph{shifting} the tokens
along the cycle $C$.
If an edge $e\in E$ is not spanned by any cycle in $\calC$, 
$e$ plays no role in shifting tokens. 
Therefore, without loss of generality, we assume that 
every edge is spanned by at least one cycle in $\calC$.

We say that two token placements $f_1$ and $f_2$ are \emph{compatible} if, for each color $c'\in\Col$, we have $\msize{f_1^{-1}(c')}=\msize{f_2^{-1}(c')}$. Obviously, compatibility is an equivalence relation on token placements, and its equivalence classes are called \emph{compatibility classes} for $G$ and $\Col$. For a compatibility class $P$ and a cycle set $\calC$, we define the \emph{token-shifting graph of $P$ and $\calC$} as the undirected graph with vertex set $P$, where there is an edge between two token placements if and only if they are adjacent with respect to $\calC$. A walk in a token-shifting graph starting from $f$ and ending in $f'$ is called a \emph{shifting sequence between $f$ and $f'$}, and the distance between $f$ and $f'$, i.e., the length of a shortest walk between them, is denoted as $\dist(f,f')$ (if there is no walk between $f$ and $f'$, their distance is defined to be $\infty$). If $\dist(f,f')<\infty$, we write $f\simeq f'$.

For a given number of colors $c$, we define the \emph{$c$-Colored Token Shift} problem as follows. The input is a graph $G=(V,E)$, a cycle set $\calC$ for $G$, two compatible token placements $f_0$ and $f_t$ (with colors drawn from the set $\Col=\{1,2,\ldots,c\}$), 
and a non-negative integer $\ell$. The goal is to determine whether
$\dist(f_0,f_t)\leq\ell$ holds. In the case that $\ell$ is not given, 
we consider the $c$-Colored Token Shift problem as 
an optimization problem that aims at computing $\dist(f_0,f_t)$.

%% file: puz.tex
For the purpose of this section, the vertex set of the graph $G=(V,E)$ will be $V=\{1,2,\dots,n\}$, and the number of colors will be $c=n$, so that $\Col=V$, and a token placement on $G$ can be interpreted as a permutation of $V$. To denote a permutation $\pi$ of $V$, we can either use the one-line notation $\pi=[\pi(1)\ \pi(2)\ \dots\ \pi(n)]$, or we can write down its cycle decomposition: for instance, the permutation $[3\ 6\ 4\ 1\ 7\ 2\ 5]$ can be expressed as the product of disjoint cycles $(1\ 3\ 4)(2\ 6)(5\ 7)$.

Note that, given a cycle set $\calC$, shifting tokens along a cycle $(v_1, v_2, \dots, v_j)\in \calC$ corresponds to applying the permutation $(v_1\ v_2\ \dots\ v_j)$ or its inverse $(v_j\ v_{j-1}\ \dots\ v_1)$ to $V$. The set of token placements generated by shifting sequences starting from the ``identity token placement'' $f_0=[1\ 2\ \dots\ n]$ is therefore a permutation group with the composition operator, which we denote by $H_{\calC}$, and we call it \emph{configuration group generated by $\calC$}. Since we visualize permutations as functions mapping vertices of $G$ to colors (and not the other way around), it makes sense to compose chains of permutations from right to left, contrary to the common convention in the permutation group literature. So, for example, if we start from the identity token placement for $n=5$ and we shift tokens along the cycles $(1\ 2\ 3)$ and $(3\ 4\ 5)$ in this order, we obtain the token placement
$$(1\ 2\ 3)(3\ 4\ 5) = [2\ 3\ 1\ 4\ 5]\, [1\ 2\ 4\ 5\ 3] = [2\ 3\ 4\ 5\ 1] = (1\ 2\ 3\ 4\ 5).$$
(Had we composed permutations from left to right, we would have obtained the token placement $[2\ 4\ 1\ 5\ 3] = (1\ 2\ 4\ 5\ 3)$ as a result.)

One of our goals in this section is to determine the configuration groups $H_{\calC}$ generated by some
classes of cycle sets $\calC$. Our choice of $\calC$ will be inspired by the puzzles in \cref{fig:puzzle},
and will consist of arrangements of cycles that share either one or two adjacent vertices. As we will see,
except in one special case, the configuration groups that we obtain are either the symmetric group $S_n$ (i.e., the group of all permutations) or the alternating group $A_n$ (i.e., the group of all even permutations), depending on whether the cycle 
set $\calC$ contains at least one even-length cycle or not: indeed, observe that a cycle of length $j$ corresponds to an even permutation if and only if $j$ is odd.

Note that the set of permutations in the configuration group $H_{\calC}$ coincides with the connected component of the token-shifting graph (as defined in the previous section) that contains $f_0$. The other connected components are simply given by the cosets of $H_{\calC}$ in $S_n$ (thus, they all have the same size), while the number of connected components of the token-shifting graph is equal to the index of $H_{\calC}$ in $S_n$, i.e., $n!/\msize{H_{\calC}}$.

The other goal of this section is to estimate the diameter of the token-shifting graph, i.e., the maximum distance between any two token placements $f_0$ and $f_t$ such that $f_0\simeq f_t$. To this end, we state some basic preliminary facts, which are folklore, and can be proved by mimicking the ``bubble sort'' algorithm.

\begin{proposition}\label{p:basic}
~\begin{enumerate}
\item[{\bf 1.}] The $n$-cycle $(1\ 2\ \dots\ n)$ and the transposition $(1\ 2)$ can generate any permutation of $\{1,2,\dots,n\}$ in $O(n^2)$ shifts.
\item[{\bf 2.}] The $n$-cycle $(1\ 2\ \dots\ n)$ and the 3-cycle $(1\ 2\ 3)$ can generate any even permutation of $\{1,2,\dots,n\}$ in $O(n^2)$ shifts.\footnote{Of course, the two cycles generate strictly more than $A_n$ (hence $S_n$) if and only if $n$ is even; however, we will only apply \cref{p:basic}.2 to generate even permutations.}
\item[{\bf 3.}] The 3-cycles $(1\ 2\ 3)$, $(2\ 3\ 4)$, \dots, $(n-2\ n-1\ n)$ can generate any even permutation of $\{1,2,\dots,n\}$ in $O(n^2)$ shifts.\qed
\end{enumerate}
\end{proposition}
All upper bounds given in \cref{p:basic} are worst-case asymptotically optimal (refer to~\cite{jerrum} for some proofs).

\subsection{Puzzles with two cycles}\label{s:3.1}

We first investigate the case where the cycle 
set $\calC$ contains exactly two cycles $\mA$ and $\mB$, either of the form $\mA=(1\ 2\ \dots\ a)$ and $\mB=(a\ a+1\ \dots\ n)$ with $1< a< n$, or of the form $\mA=(1\ 2\ \dots\ a)$ and $\mB=(a-1\ a\ a+1\ \dots\ n)$, with $1< a\leq n$. The first puzzle is called \emph{1-connected $(a,b)$-puzzle}, where $n=a+b-1$, and the second one is called \emph{2-connected $(a,b)$-puzzle}, where $n=a+b-2$ (so, in both cases $a>1$ and $b>1$ are the lengths of the two cycles $\mA$ and $\mB$, respectively). See \cref{fig:2puzzle} for some examples. Note that the Turnstile puzzle in \cref{fig:puzzle} (left) can be regarded as a 2-connected $(6,6)$-puzzle.

\begin{figure}
\centering
\includegraphics[width=\textwidth]{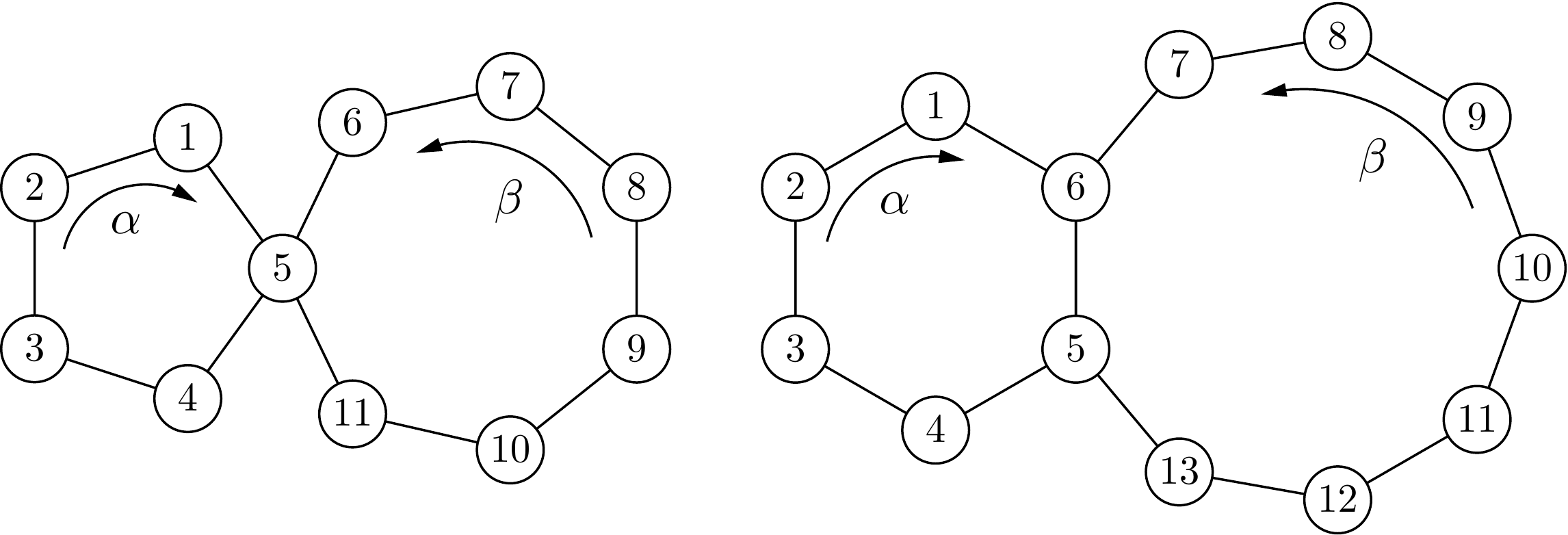}
\caption{A 1-connected $(5,7)$-puzzle (left) and a 2-connected $(6,9)$-puzzle (right)}
\label{fig:2puzzle}
\end{figure}

\begin{theorem}\label{1-ab-puz}
The configuration group of a 1-connected $(a,b)$-puzzle is $A_n$ if both $a$ and $b$ are odd, and it is $S_n$ otherwise. Any permutation in the configuration group can be generated in $O(n^2)$ shifts.
\end{theorem}
\begin{proof}
Observe that the commutator of $\mA$ and $\mBi$ is the 3-cycle $\mAi\mB\mA\mBi=(a-1\ a\ a+1)$. So, we can apply \cref{p:basic}.2 to the $n$-cycle $\mA\mB=(1\ 2\ \dots\ n)$ and the 3-cycle $(a-1\ a\ a+1)$ to generate any even permutation in $O(n^2)$ shifts. If $a$ and $b$ are odd, then $\mA$ and $\mB$ are even permutations, and therefore cannot generate any odd permutation.

On the other hand, if $a$ is even (the case where $b$ is even is symmetric), then the $a$-cycle $\mA$ is an odd permutation. So, to generate any odd permutation $\pi\in S_n$, we first generate the even permutation $\pi\mA$ in $O(n^2)$ shifts, and then we do one extra shift along the cycle $\mAi$.\qed
\end{proof}

Our first observation about 2-connected $(a,b)$-puzzles is that the composition of $\mAi$ and $\mB$ is the $(n-1)$-cycle $\mAi\mB=(a-2\ a-3\ \dots\ 1\ a\ a+1\ \dots\ n)$, which excludes only the element $a-1$. Similarly, $\mA\mBi=(1\ 2\ \dots\ a-1\ n\ n-1\ \dots\ a+1)$, which excludes only the element $a$. We will write $\mCo$ and $\mCt$ as shorthand for $\mAi\mB$ and $\mA\mBi$ respectively, and we will use the permutations $\mCo$ and $\mCt$ to conjugate $\mA$ and $\mB$, thus obtaining different $a$-cycles and $b$-cycles.\footnote{If $g$ and $h$ are two elements of a group, the \emph{conjugate} of $g$ by $h$ is defined as $h^{-1}gh$. In the context of permutation groups, conjugation by any $h$ is an automorphism that preserves the cycle structure of permutations~\cite[Theorem~3.5]{rotman}.}

\begin{lemma}\label{2-3b-puz}
In a 2-connected $(3,b)$-puzzle, any even permutation can be generated in $O(n^2)$ shifts.
\end{lemma}
\begin{proof}
If we conjugate the 3-cycle $\mAi$ by the inverse of $\mCo$, we obtain the 3-cycle $\mCo\mAi\mCoi=(2\ 3\ 4)$. By applying \cref{p:basic}.2 to the $(n-1)$-cycle $\mB$ and the 3-cycle $(2\ 3\ 4)$, we can generate any even permutation of $V\setminus\{1\}$ in $O(n^2)$ shifts.

Let $\pi\in A_n$ be an even permutation of $V$. In order to generate $\pi$, we first move the correct token $\pi(1)$ to position~1 in $O(n)$ shifts, possibly scrambling the rest of the tokens: let $\sigma$ be the resulting permutation. If $\sigma$ is even, then $\sigma^{-1}\pi$ is an even permutation of $V\setminus\{1\}$, and we can generate it in $O(n^2)$ shifts as shown before, obtaining $\pi$.

On the other hand, if $\sigma$ is odd, then one of the generators $\mA$ and $\mB$ must be odd, too. Since $\mA$ is a 3-cycle, it follows that $\mB$ is odd. In this case, after placing the correct token in position~1 via $\sigma$, we shift the rest of the tokens along $\mB$, and then we follow up with $\mBi\sigma^{-1}\pi$, which is an even permutation of $V\setminus\{1\}$, and can be generated it in $O(n^2)$ shifts. Again, the result is $\sigma\mB\mBi\sigma^{-1}\pi=\pi$.\qed
\end{proof}

\begin{lemma}\label{2-45-puz}
In a 2-connected $(a,b)$-puzzle with $a\geq 4$ and $b\geq 5$, any even permutation can be generated in $O(n^2)$ shifts.
\end{lemma}
\begin{proof}
As shown in \cref{fig:2-45-puz}, the conjugate of $\mB$ by $\mCo$ is the $b$-cycle $$\mDo=\mCoi\mB\mCo=(a\ a+1\ \dots\ n-1\ a-1\ 1),$$ and the conjugate of $\mBi$ by $\mCt$ is the $b$-cycle $$\mDt=\mCti\mBi\mCt=(n\ n-1\ \dots\ a+2\ a\ a-2\ a-1).$$ Their composition is $\mDo\mDt=(1\ a\ a-2)(a-1\ n)(a+1\ a+2)$, and therefore $(\mDo\mDt)^2$ is the 3-cycle $(1\ a-2\ a)$. Conjugating this 3-cycle by $\mAi$, we finally obtain the 3-cycle $\mT=\mA(\mDo\mDt)^2\mAi=(1\ 2\ a-1)$; note that $\mT$ has been generated in a number of shifts independent of $n$. Now, since the 3-cycle $\mT$ and the $(n-1)$-cycle $\mCt$ induce a 2-connected $(3,n-1)$-puzzle on $V$, we can apply \cref{2-3b-puz} to generate any even permutation of $V$ in $O(n^2)$ shifts.\qed
\end{proof}

\begin{figure}
\centering
\includegraphics[width=\textwidth]{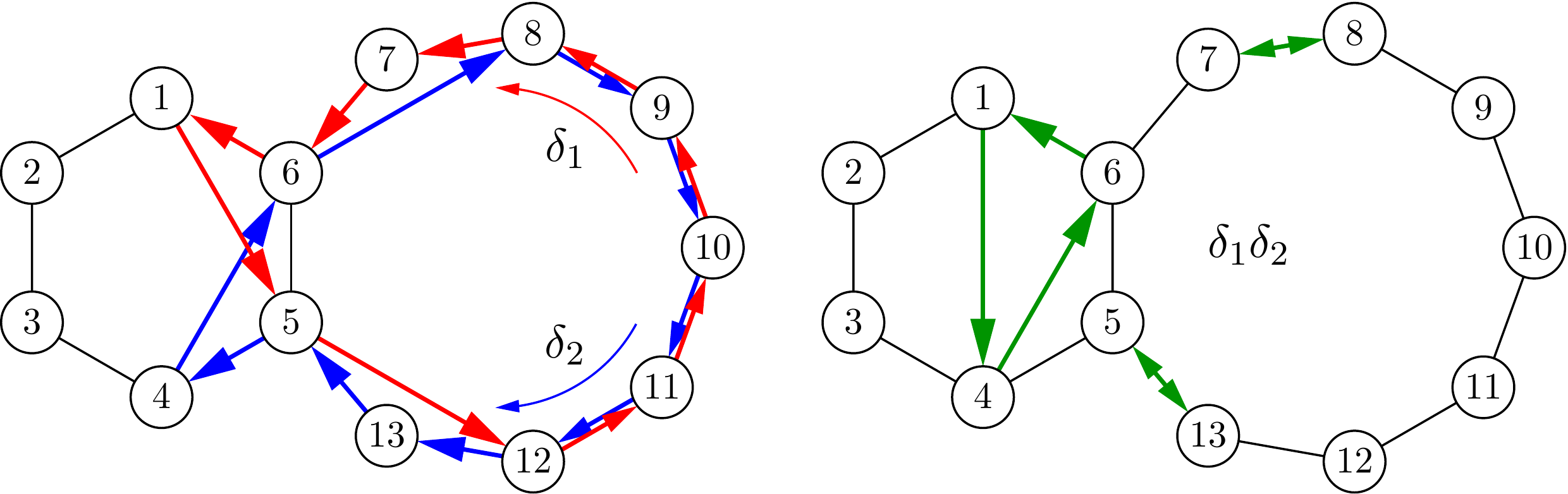}
\caption{Some permutations constructed in the proof of \cref{2-45-puz}}
\label{fig:2-45-puz}
\end{figure}

\begin{theorem}\label{2-ab-puz}
The configuration group of a 2-connected $(a,b)$-puzzle is:
\begin{enumerate}
\item[{\bf 1.}] Isomorphic to $S_{n-1}=S_5$ if $a=b=4$.
\item[{\bf 2.}] $A_n$ if both $a$ and $b$ are odd.
\item[{\bf 3.}] $S_n$ otherwise.
\end{enumerate}
Any permutation in the configuration group can be generated in $O(n^2)$ shifts.
\end{theorem}
\begin{proof}
By the symmetry of the puzzle, we may assume $a\leq b$. The case with $a=2$ is equivalent to \cref{p:basic}.1, so let $a\geq 3$. If $a\neq 4$ or $b\neq 4$, then \cref{2-3b-puz,2-45-puz} apply, hence we can generate any even permutation in $O(n^2)$ shifts: the configuration group is therefore at least $A_n$. Now we reason as in \cref{1-ab-puz}: if $a$ and $b$ are odd, then $\mA$ and $\mB$ are even permutations, and cannot generate any odd one. If $a$ is even (the case where $b$ is even is symmetric), then $\mA$ is an odd permutation. In this case, to generate any odd permutation $\pi\in S_n$, we first generate the even permutation $\pi\mA$ in $O(n^2)$ shifts, and then we do one more shift along the cycle $\mAi$ to obtain $\pi$.

The only case left is $a=b=4$. To analyze the 2-connected $(4,4)$-puzzle, consider the outer automorphism $\psi\colon S_6\to S_6$ defined on a generating set of $S_6$ as follows (cf.~\cite[Corollary~7.13]{rotman}):
\begin{align*}
\psi((1\ 2))&=(1\ 5)(2\ 3)(4\ 6),\qquad \psi((1\ 3))=(1\ 4)(2\ 6)(3\ 5),\\
\psi((1\ 4))&=(1\ 3)(2\ 4)(5\ 6),\qquad \psi((1\ 5))=(1\ 2)(3\ 6)(4\ 5),\\
\psi((1\ 6))&=(1\ 6)(2\ 5)(3\ 4).
\end{align*}
Because $\psi$ is an automorphism, the subgroup of $S_6$ generated by $\mA$ and $\mB$ is isomorphic to the subgroup generated by the permutations $\psi(\mA)$ and $\psi(\mB)$. Since $\mA=(1\ 2\ 3\ 4)=(1\ 2)(1\ 3)(1\ 4)$ and $\mB=(3\ 4\ 5\ 6)=(1\ 3)(1\ 4)(1\ 5)(1\ 6)(1\ 3)$, and recalling that $\psi(\pi_1\pi_2)=\psi(\pi_1)\psi(\pi_2)$ for all $\pi_1,\pi_2\in S_6$, we have:
\begin{align*}
\psi(\mA)&=\psi((1\ 2))\psi((1\ 3))\psi((1\ 4))=[1\ 5\ 6\ 4\ 3\ 2]=(2\ 5\ 3\ 6) \mbox{ and}\\
\psi(\mB)&=\psi((1\ 3))\psi((1\ 4))\psi((1\ 5))\psi((1\ 6))\psi((1\ 3))=[3\ 1\ 5\ 4\ 2\ 6]=(1\ 3\ 5\ 2).
\end{align*}
Note that the new generators $\psi(\mA)$ and $\psi(\mB)$ both leave the token $4$ in place, and so they cannot generate a subgroup larger than $S_5$ (up to isomorphism). On the other hand, we have $\psi(\mA)\psi(\mB)=(1\ 6\ 2)$. This 3-cycle, together with the 4-cycle $\psi(\mA)$, induces a 2-connected $(3,4)$-puzzle on $\{1,2,3,5,6\}$: as shown before, the configuration group of this puzzle is (isomorphic to) $S_5$. We conclude that the configuration group of the 2-connected $(4,4)$-puzzle is isomorphic to $S_5$, as well. A given permutation $\pi\in S_6$ is in the configuration group if and only if $\psi(\pi)$ leaves the token $4$ in place.\qed
\end{proof}

\subsection{Puzzles with any number of cycles}
Let us generalize the $(a,b)$-puzzle to larger numbers of cycles.
(As far as the authors know, there are commercial products that have 2, 3, 4, and 6 cycles.)
We say that two cycles are \emph{properly interconnected} if they share exactly one vertex, of if they share exactly two vertices which are consecutive in both cycles. Note that all 1-connected and 2-connected $(a,b)$-puzzles consist of two properly interconnected cycles. Given a set of cycles $\calC$ in a graph $G=(V,E)$, let us define the \emph{interconnection graph} $\hat G=(\calC, \hat E)$, where there is an (undirected) edge between two cycles of $\calC$ if and only if they are properly interconnected.

Let us assume $\msize{V}>6$ (to avoid special configurations of small size, which can be analyzed by hand), and let $\calC$ consist of $k$ cycles of lengths $n_1$, $n_2$, \dots, $n_k$, respectively. We say that $\calC$ induces a \emph{generalized $(n_1, n_2, \dots, n_k)$-puzzle} on $V$ if there is a subset $\calC'\subseteq \calC$ such that:
\begin{enumerate}
\item[(1)] $\calC'$ contains at least two cycles;
\item[(2)] the induced subgraph $\hat G[\calC']$ is connected;
\item[(3)] each vertex of $G$ is contained in at least one cycle in $\calC'$.
\end{enumerate}
When we fix such a subset $\calC'$, the cycles in $\calC'$ are called \emph{relevant cycles}, and the vertices of $G$ that are shared by two properly interconnected relevant cycles are called \emph{relevant vertices} for those cycles. See \cref{fig:more-cycles} for an example of a generalized puzzle.

\begin{figure}
\centering
\includegraphics[width=\textwidth]{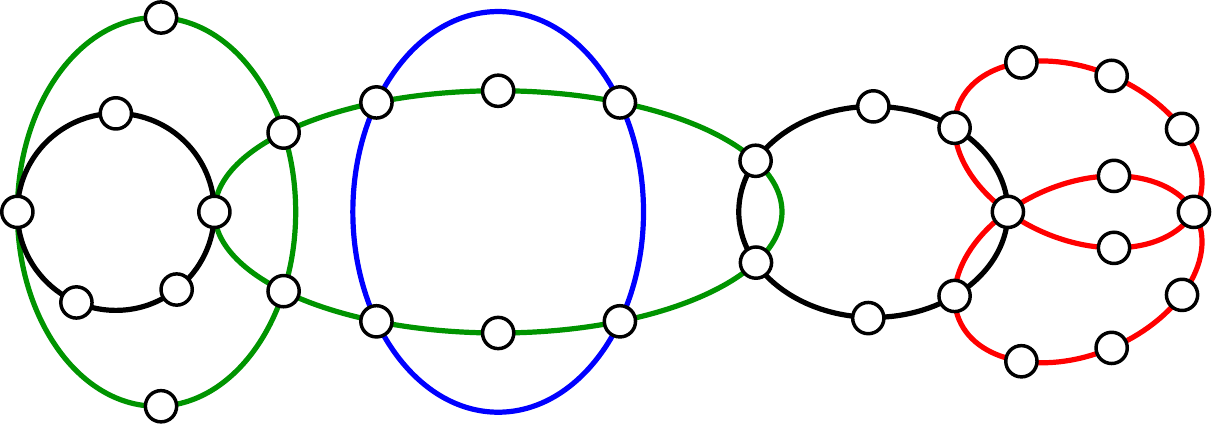}
\caption{A generalized puzzle where any permutation can be generated in $O(n^5)$ shifts, due to \cref{gen-puz}. Note that the blue cycle is the only cycle of even length, and is not properly interconnected with any other cycle. Also, the two red cycles and the two green cycles intersect each other but are not properly interconnected.}
\label{fig:more-cycles}
\end{figure}

The next two lemmas are technical; their proof is found in the Appendix.

\begin{lemma}\label{abc-puz}
In a generalized puzzle with three relevant cycles, $\calC'=\{C_1,C_2,C_3\}$, such that $C_1$ and $C_2$ induce a 2-connected $(4,4)$-puzzle, any permutation involving only vertices in $C_1$ and $C_2$ can be generated in $O(n^2)$ shifts.\qed
\end{lemma}

\begin{lemma}\label{pathlemma}
Let $V=\{1,\dots,n\}$, and let $W=(w_1,\dots,w_m)\in V^m$ be a sequence such that each element of $V$ appears in $W$ at least once, and any three consecutive elements of $W$ are distinct. Then, the set of 3-cycles $\calC=\{(w_{i-1}\ w_{i}\ w_{i+1})\mid 1< i< m\}$ can generate any even permutation of $V$ in $O(n^3)$ shifts.\qed
\end{lemma}

\begin{theorem}\label{gen-puz}
The configuration group of a generalized $(n_1,n_2,\dots,n_k)$-puzzle is $A_n$ if $n_1$, $n_2$, \dots, $n_k$ are all odd, and it is $S_n$ otherwise. Any permutation in the configuration group can be generated in $O(n^5)$ shifts.
\end{theorem}
\begin{proof}
Observe that it suffices to prove that the given cycles can generate any even permutation in $O(n^5)$ shifts. Indeed, if all cycles have odd length, they cannot generate any odd permutation. On the other hand, if there is a cycle of even length and we want to generate an odd permutation $\pi$, we can shift tokens along that cycle, obtaining an odd permutation $\sigma$, and then we can generate the even permutation $\sigma^{-1}\pi$ in $O(n^5)$ shifts, obtaining $\pi$.

Let us fix a set of $k'\geq 3$ relevant cycles $\calC'\subseteq \calC$: we will show how to generate any even permutation by shifting tokens only along relevant cycles. By properties~(2) and~(3) of generalized puzzles, there exists a walk $W$ on $G$ that visits all vertices (possibly more than once), traverses only edges of relevant cycles, and transitions from one relevant cycle to another only if they are properly interconnected, and only through a relevant vertex shared by them. We will now slightly modify $W$ so that it satisfies the hypotheses of \cref{pathlemma}, as well as some other conditions. Namely, if $w_{i-1}$, $w_{i}$, $w_{i+1}$ are any three vertices that are consecutive in $W$, we would like the following conditions to hold:
\begin{enumerate}
\item[(1)] $w_{i-1}$, $w_{i}$, $w_{i+1}$ are all distinct (this is the condition required by \cref{pathlemma});
\item[(2)] either $w_{i-1}$ and $w_{i}$ are in the same relevant cycle, or $w_{i}$ and $w_{i+1}$ are in the same relevant cycle;
\item[(3)] $w_{i-1}$ and $w_{i+1}$ are either in the same relevant cycle, or in two properly interconnected relevant cycles.
\end{enumerate}
To satisfy all conditions, it is sufficient to let $W$ do a whole loop around a relevant cycle before transitioning to the next (note that \cref{pathlemma} applies regardless of the length of $W$). The only case where this is not possible is when $W$ has to go through a relevant 2-cycle $C=(u_1\ u_2)$ that is a leaf in the induced subgraph $\hat G[\calC']$, such that $C$ shares exactly one relevant vertex, say $u_{1}$, with another relevant cycle $C'=(v_0\ u_{1}\ v_1\ v_2\ \dots)$. To let $W$ cover $C$ in a way that satisfies the above conditions, we set either $W=(\dots, v_0, u_1, u_2, v_1,\dots)$ or $W=(\dots, v_1, u_1, u_2, v_0,\dots)$: that is, we skip $u_1$ after visiting $u_2$. After this modification, $W$ is no longer a walk on $G$, but it satisfies the hypotheses of \cref{pathlemma}, as well as the three conditions above.

We will now show that the 3-cycle $(w_{i-1}\ w_{i}\ w_{i+1})$ can be generated in $O(n^2)$ shifts, for all $1<i<\msize{W}$. By \cref{pathlemma}, we will therefore conclude that any even permutation of $V$ can be generated in $O(n^2)\cdot O(n^3)=O(n^5)$ shifts. Due to conditions~(2) and~(3), we can assume without loss of generality that $w_{i-1}$ and $w_{i}$ are both in the same relevant cycle $C_1$, and that $w_{i+1}$ is either in $C_1$ or in a different relevant cycle $C_2$ which is properly interconnected with $C_1$. In the first case, by property~(1) of generalized puzzles, there exists another relevant cycle $C_2$ properly interconnected with $C_1$. So, in all cases, $C_1$ and $C_2$ induce a 1-connected or a 2-connected $(\msize{C_1},\msize{C_2})$-puzzle.

That the 3-cycle $(w_{i-1}\ w_{i}\ w_{i+1})$ can be generated in $O(n^2)$ shifts now follows directly from \cref{1-ab-puz,2-ab-puz}, except if $\msize{C_1}=\msize{C_2}=4$ and $C_1$ and $C_2$ share exactly two vertices: indeed, the 2-connected $(4,4)$-puzzle is the only case where we cannot generate any 3-cycle. However, since we are assuming that $V>6$, there must be a third relevant cycle $C_3$, which is properly interconnected with $C_1$ or $C_2$. Our claim now follows from \cref{abc-puz}.\qed
\end{proof}

%% file: np.tex
In this section, we show that the 2-Colored Token Shift problem is NP-hard.
That is, for a graph $G=(V,E)$, cycle set $\calC$, two token placements $f_0$ and $f_t$ for $G$, and a non-negative integer $\ell$, it is NP-hard to determine if $\dist(f_0,f_t)\le \ell$.
\begin{theorem}
\label{th:2color}
The 2-Colored Token Shift problem is NP-hard.
\end{theorem}

\begin{proof}
We will give a polynomial-time reduction from the NP-complete problem 3-Dimensional Matching, or 3DM~\cite{garey1979computers}: given three disjoint sets $X$, $Y$, $Z$, each of size $m$, and a set of triplets $T\subseteq X\times Y\times Z$, does $T$ contain a matching, i.e., a subset $M\subseteq T$ of size exactly $m$ such that all elements of $X$, $Y$, $Z$ appear in $M$?

Given an instance of 3DM $(X,Y,Z,T)$, with $n=\msize{T}$, we construct the instance of the 2-Colored Token Shift problem illustrated in \cref{fig:reduction}.

\begin{figure}
\centering
\includegraphics[width=\textwidth]{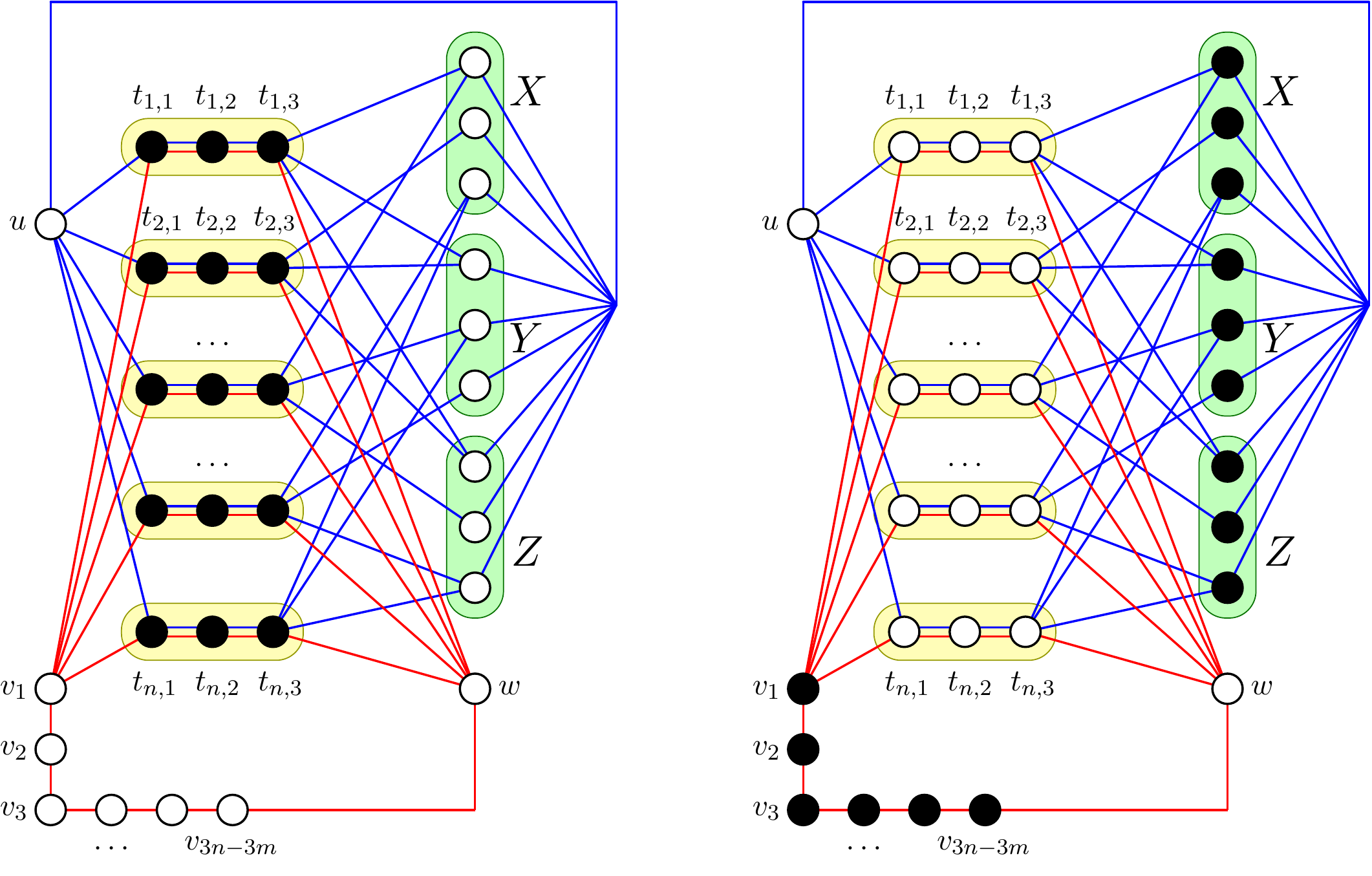}
\caption{The initial token placement $f_0$ (left) and the final token placement $f_t$ (right)}
\label{fig:reduction}
\end{figure}

The vertex set of $G=(V,E)$ includes the sets $X$, $Y$, $Z$ (shown with a green background in the figure: these will be called \emph{green vertices}), as well as the vertex $u$. Also, for each triplet $\hat t_i=(x,y,z)\in T$, with $1\leq i\leq n$, the vertex set contains three vertices $t_{i,1}$, $t_{i,2}$, $t_{i,3}$ (shown with a yellow background in the figure: these will be called \emph{yellow vertices}), and the cycle set $\calC$ has the three cycles $(u, t_{i,1}, t_{i,2}, t_{i,3}, x)$, $(u, t_{i,1}, t_{i,2}, t_{i,3}, y)$, and $(u, t_{i,1}, t_{i,2}, t_{i,3}, z)$ (drawn in blue in the figure). Finally, we have the vertex $w$, and the vertices $v_1$, $v_2$, \dots, $v_{3n-3m}$; for each $i\in \{1,2,\dots, n\}$, the cycle set $\calC$ contains the cycle $(t_{i,3}, t_{i,2}, t_{i,1}, v_1, v_2, \dots, v_{3n-3m}, w)$ (drawn in red in the figure). In the initial token placement $f_0$, there are black tokens on the $3n$ vertices of the form $t_{i,j}$, and white tokens on all other vertices. In the final token placement $f_t$, there is a total of $3m$ black tokens on all the vertices in $X$, $Y$, $Z$, plus $3n-3m$ black tokens on $v_1$, $v_2$, \dots, $v_{3n-3m}$; all other vertices have white tokens. With this setup, we let $\ell=3n$.

It is easy to see that, if the 3DM instance has a matching $M=\{\hat t_{i_1},\hat t_{i_2},\dots, \hat t_{i_m}\}$, then $\dist(f_0,f_t)\leq \ell$. Indeed, for each $\hat t_{i_j}=(x_j,y_j,z_j)$, with $1\leq j\leq m$, we can shift tokens along the three blue cycles containing the yellow vertices $t_{i_j,1}$, $t_{i_j,2}$, $t_{i_j,3}$, thus moving their three black tokens into the green vertices $x_j$, $y_j$, and $z_j$. Since $M$ is a matching, these $3m$ shifts eventually result in $X$, $Y$, and $Z$ being covered by black tokens. Finally, we can shift the $3n-3m$ black tokens corresponding to triplets in $T\setminus M$ along red cycles, moving them into the vertices $v_1$, $v_2$, \dots, $v_{3n-3m}$. Clearly, this is a shifting sequence of length $3n=\ell$ from $f_0$ to $f_t$.

We will now prove that, assuming that $\dist(f_0,f_t)\leq\ell$, the 3DM instance has a matching. Note that each shift, no matter along which cycle, can move at most one black token from a yellow vertex to a non-yellow vertex. Since in $f_0$ there are $\ell=3n$ black tokens on yellow vertices, and in $f_t$ no token is on a yellow vertex, it follows that each shift must cause exactly one black token to move from a yellow vertex to a non-yellow vertex, and no black token to move back into a yellow vertex.

This implies that no black token should ever reach vertex $u$: if it did, it would eventually have to be moved to some other location, because $u$ does not hold a black token in $f_t$. However, the black token in $u$ cannot be shifted back into a yellow vertex, and therefore it will be shifted into a green vertex along a blue cycle. Since every shift must cause a black token to leave the set of yellow vertices, such a token will move into $u$: we conclude that $u$ will always contain a black token, which is a contradiction. Similarly, we can argue that the vertex $w$ should never hold a black token.

Let us now focus on a single triplet of yellow vertices $t_{i,1}$, $t_{i,2}$, $t_{i,3}$. Exactly three shifts must involve these vertices, and they must result in the three black tokens leaving such vertices. Clearly, this is only possible if the three black tokens are shifted in the same direction. If they are shifted in the direction of $t_{i,3}$ (i.e., rightward in \cref{fig:reduction}), they must move into green vertices (because they cannot go into $w$); if they are shifted in the direction of $t_{i,1}$ (i.e., leftward in \cref{fig:reduction}), they must move into $v_1$ (because they cannot go into $u$).

Note that, if a black token ever reaches a green vertex, it can no longer be moved: any shift involving such a token would move it back into a yellow vertex or into $u$. It follows that the only way of filling all the green vertices with black token is to select a subset of exactly $m$ triplets of yellow vertices and shift each of their black tokens into a different green vertex. These $m$ triplets of yellow vertices correspond to a matching for the 3DM instance.\qed
\end{proof}

In the above reduction, we can easily observe that the final token placement $f_t$ can always be reached from the initial token placement $f_0$ in a polynomial number of shifts. Therefore, for this particular set of instances, the 2-Colored Token Shift problem is in NP. The same is also true of the puzzles introduced in \cref{sec:3}, due to the polynomal upper bound given by \cref{gen-puz}. However, we do not know whether this is true for the $c$-Colored Token Shift problem in general, even assuming $c=2$. A theorem of Helfgott and Seress~\cite{helfgott} implies that, if $f_0\simeq f_t$, the distance between $f_0$ and $f_t$ has a quasi-polynomial upper bound; this, however, is insufficient to conclude that the problem is in NP. On the other hand, it is not difficult to see that the $c$-Colored Token Shift problem is in PSPACE; characterizing its computational complexity is left as an open problem. It would also be interesting to establish if the problem remains NP-hard when restricted to planar graphs or to graphs of constant maximum degree.

%% file: appendix.tex
\newpage

\section*{Appendix}
\subsection*{Additional Figures}
\begin{figure}
\centering
\includegraphics[height=4.35cm]{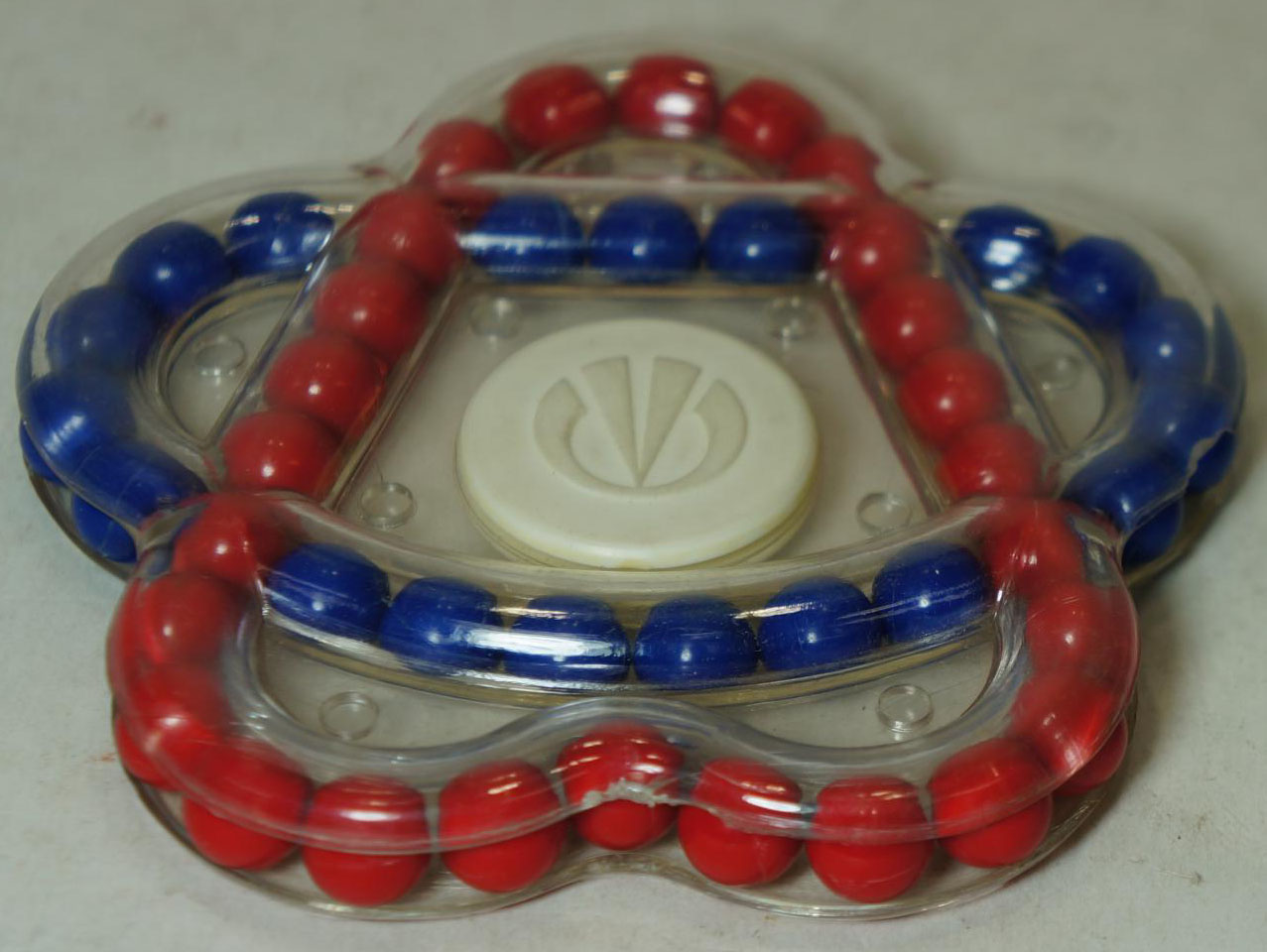}
\hspace{5mm}
\includegraphics[height=4.35cm]{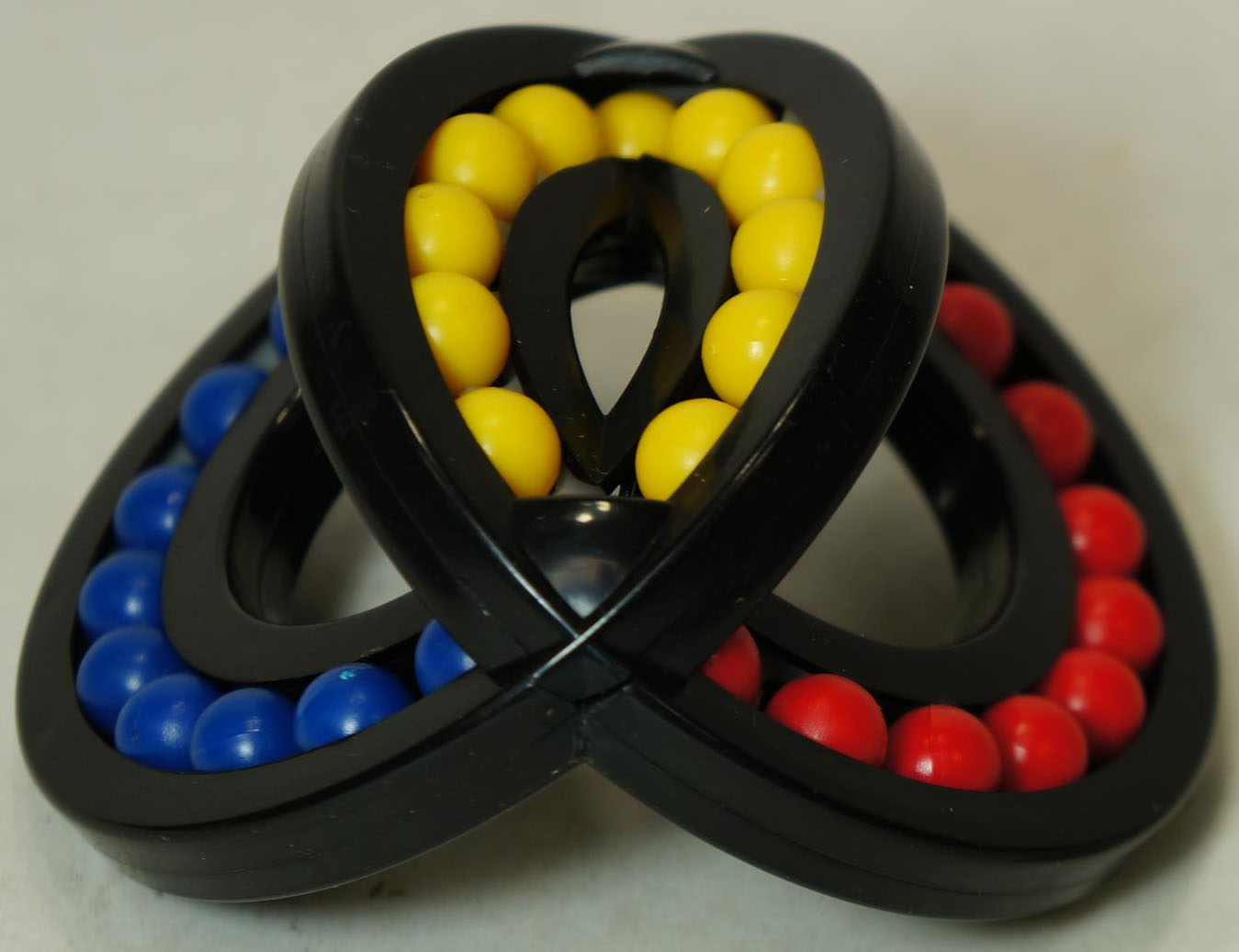}
\caption{Some cyclic shift puzzles with two (not properly interconnected) cycles}
\end{figure}

\begin{figure}
\centering
\includegraphics[height=4.5cm]{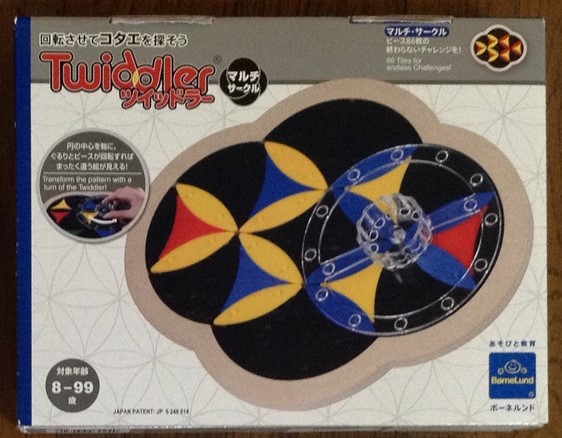}
\hspace{5mm}
\includegraphics[height=4.5cm]{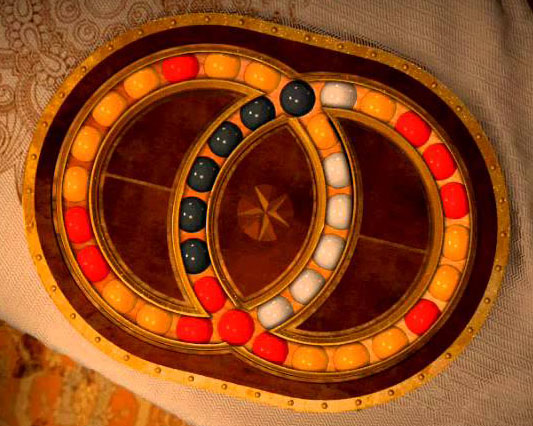}
\caption{More cyclic shift puzzles: Twiddler (left) and a puzzle found in the video game Haunted Manor 2 (right)}
\end{figure}

\subsection*{Missing Proofs}
\setcounter{lemma}{2}
\begin{lemma}
In a generalized puzzle with three relevant cycles, $\calC'=\{C_1,C_2,C_3\}$, such that $C_1$ and $C_2$ induce a 2-connected $(4,4)$-puzzle, any permutation involving only vertices in $C_1$ and $C_2$ can be generated in $O(n^2)$ shifts.
\end{lemma}
\begin{proof}
Let $\alpha=(1\ 2\ 3\ 4)$ and $\beta=(3\ 4\ 5\ 6)$ be the permutations corresponding to shifting tokens along $C_1$ and $C_2$, respectively. As in \cref{s:3.1}, we set $\mCo=\mAi\mB=(1\ 4\ 5\ 6\ 2)$ and $\mCt=\mA\mBi=(1\ 2\ 3\ 6\ 5)$. Since we are assuming that $\msize{V}>6$, there must be a seventh vertex, and shifting along $C_3$ corresponds to a permutation of the form $\tau=(\dots\ 7\ \dots)$.

We will prove that it is always possible to generate a transposition of the form $(3\ x)$, with $x\in\{1,2,4,5,6\}$, in $O(n^2)$ shifts. Indeed, such a transposition, together with the 5-cycle $\mCo$, induces a 1-connected $(2,5)$-puzzle on $\{1,2,3,4,5,6\}$. Our lemma will thus follow from \cref{1-ab-puz} and the fact that, in a 1-connected $(2,5)$-puzzle, the distance between any two token placements is bounded by a constant.

If $\msize{C_3}\neq 4$, or if $C_3$ is 1-connected with $C_1$ or $C_2$, then the transposition $(3\ 4)$ can be generated in $O(n^2)$ shifts, due to \cref{1-ab-puz,2-ab-puz}. So, we may assume that $\msize{C_3}=4$, and $C_3$ is properly interconnected with $C_2$ and shares exactly two vertices with it. Perhaps, $C_3$ shares at least two vertices with $C_1$, as well. The only possible configurations, up to symmetry, are the following:
\begin{enumerate}
\item[(1)] $\tau=(3\ 4\ 7\ 8)$. Then, $\tau$ and $\gamma_1$ induce a 1-connected $(4,5)$-puzzle on $V$, and can generate the transposition $(3\ 4)$ by \cref{1-ab-puz}.
\item[(2)] $\tau=(5\ 6\ 7\ 8)$. Then, $\tau$ and $\gamma_2$ induce a 2-connected $(4,5)$-puzzle on $V\setminus\{4\}$, and can generate the transposition $(3\ 5)$ by \cref{2-ab-puz}.
\item[(3)] $\tau=(1\ 7\ 3\ 4)$. In this case, $(3\ 2) = \tau^{-2}\alpha\tau\alpha$.
\item[(4)] $\tau=(1\ 3\ 4\ 7)$. In this case, $(3\ 4) = \alpha\beta^{-1}\alpha^{-1}\tau\beta\tau^2$.
\item[(5)] $\tau=(1\ 3\ 6\ 7)$. In this case, $(3\ 5) = \beta^{-1}\alpha\tau^{-1}\alpha\beta\tau^2$.
\item[(6)] $\tau=(1\ 6\ 3\ 7)$. In this case, $(3\ 1) = \alpha\beta\alpha^{-1}\beta\tau^{-1}\beta\tau$.
\item[(7)] $\tau=(2\ 6\ 3\ 7)$. In this case, $(3\ 4) = \alpha^2\tau^2\alpha\tau^2$.
\item[(8)] $\tau=(2\ 3\ 6\ 7)$. In this case, $(3\ 1) = \tau^{-1}\beta^{-1}\alpha\beta\alpha^{-1}\tau\alpha$.\qed
\end{enumerate}
\end{proof}

\begin{lemma}
Let $V=\{1,\dots,n\}$, and let $W=(w_1,\dots,w_m)\in V^m$ be a sequence such that each element of $V$ appears in $W$ at least once, and any three consecutive elements of $W$ are distinct. Then, the set of 3-cycles $\calC=\{(w_{i-1}\ w_{i}\ w_{i+1})\mid 1< i< m\}$ can generate any even permutation of $V$ in $O(n^3)$ shifts.
\end{lemma}
\begin{proof}
Let $\mu\colon V\to \{1,\dots,m\}$ be the function mapping each $v\in V$ to the minimum index $\mu(v)$ such that $w_{\mu(v)}=v$. Let $\pi=[\pi_1\ \dots\ \pi_n]$ be the permutation of $V$ such that the sequence $(\mu(\pi_1),\dots,\mu(\pi_n))$ is monotonically increasing.

We will prove by induction on $i$ that $\calC$ can generate any 3-cycle on $\{\pi_1,\dots,\pi_i\}$ in at most $3i$ shifts. Assume this claim to be true up to a certain $i<n$, and let us prove it for $i+1$. Let $\mathcal T=\{(\pi_{j}\ \pi_{j'}\ \pi_{i+1})\mid 1\leq j<j'\leq i\}$, and note that it suffices to prove that $\calC$ generates all 3-cycles in $\mathcal T$, because the 3-cycles on $\{\pi_1,\dots,\pi_i\}$ are already accounted for by the inductive hypothesis.

So, fix one such 3-cycle $\sigma_1=(\pi_{j}\ \pi_{j'}\ \pi_{i+1})\in \mathcal T$, and observe that $\calC$ already contains a 3-cycle in $\mathcal T$, namely $\sigma_2=(w_{\mu(\pi_{i+1})-2}\ w_{\mu(\pi_{i+1})-1}\ w_{\mu(\pi_{i+1})})$. Indeed, we have $w_{\mu(\pi_{i+1})}=\pi_{i+1}$, and, by the minimality of $\mu$, there exist two distinct indices $k,k'\in \{1,\dots, i\}$ such that $w_{\mu(\pi_{i+1})-2}=\pi_k$ and $w_{\mu(\pi_{i+1})-1}=\pi_{k'}$.

If $\{j,j'\}=\{k,k'\}$, then $\sigma_1=\sigma_2$, and we are done. If $\{j,j'\}$ and $\{k,k'\}$ are disjoint, consider the 3-cycle $\sigma_3=(\pi_j\ \pi_{j'}\ \pi_{k})$, which, by the inductive hypothesis, can be generated by $\calC$ in at most $3i$ shifts. We have $\sigma_1=\sigma_2\sigma_3\sigma_2\sigma_2$, and so $\calC$ can generate $\sigma_1$ in at most $3i+3=3(i+1)$ shifts.

Otherwise, $\{j,j'\}$ and $\{k,k'\}$ intersect in exactly one element, which we may assume to be $j'=k'$, without loss of generality. In this case, $\sigma_1=\sigma_2\sigma_3$, where $\sigma_3$ is defined as above. So, $\calC$ can generate $\sigma_1$ in at most $3i+1<3(i+1)$ shifts.

By taking $i=n$, we conclude that $\calC$ can generate any 3-cycle on $V$ in at most $3n=O(n)$ shifts, implying that it can generate any even permutation of $V$ in $O(n^3)$ shifts, due to \cref{p:basic}.3.\qed
\end{proof}

%% file: main.bbl
\begin{thebibliography}{99}

\bibitem{babai}
L\'aszl\'o Babai.
\newblock The Probability of Generating the Symmetric Group.
\newblock {\em Journal of Combinatorial Theory} (Series A), 52:148--153, 1989.

\bibitem{furst}
Merrick Furst, John Hopcroft, and Eugene Luks.
\newblock Polynomial-Time Algorithms for Permutation Groups.
\newblock In {\em Proceedings of the 21st Annual Symposium on Foundations of Computer Science}, 36--41, 1980.

\bibitem{garey1979computers}
Michael~R. Garey and David~S. Johnson.
\newblock {\em Computers and Intractability: A Guide to the Theory of NP-Completeness}.
\newblock W.~H. Freeman and Company, 1979.

\bibitem{heath}
Daniel Heath, I.~M. Isaacs, John Kiltinen, and Jessica Sklar.
\newblock Symmetric and Alternating Groups Generated by a Full Cycle and Another Element.
\newblock {\em The American Mathematical Monthly}, 116(5):447--451, 2009.

\bibitem{helfgott}
Harald~A. Helfgott and \'Akos Seress.
\newblock On the Diameter of Permutation Groups.
\newblock {\em Annals of Mathematics}, 179(2):611--658, 2014.

\bibitem{jerrum}
Mark~R. Jerrum.
\newblock The Complexity of Finding Minimum-Length Generator Sequences.
\newblock {\em Theoretical Computer Science}, 36:265--289, 1985.

\bibitem{jones}
Gareth~A. Jones.
\newblock Primitive Permutation Groups Containing a Cycle.
\newblock {\em Bulletin of the Australian Mathematical Society}, 89(1):159--165, 2014.

\bibitem{nishimura}
Naomi Nishimura.
\newblock Introduction to Reconfiguration.
\newblock {\em Algorithms}, 11(4):1--25, 2018.

\bibitem{rotman}
Joseph~J. Rotman.
\newblock {\em An Introduction to the Theory of Groups}.
\newblock Springer-Verlag, 4th edition, 1995.

\bibitem{wilson}
Richard~M. Wilson.
\newblock Graph Puzzles, Homotopy, and the Alternating Group.
\newblock {\em Journal of Combinatorial Theory} (Series B), 16:86--96, 1974.

\bibitem{Yamanakaetal2018}
Katsuhisa Yamanaka, Takashi Horiyama, J.~Mark Keil, David Kirkpatrick, Yota Otachi, Toshiki Saitoh, Ryuhei Uehara, and Yushi Uno.
\newblock Swapping Colored Tokens on Graphs.
\newblock {\em Theoretical Computer Science}, 729:1--10, 2018.

\end{thebibliography}
